\documentclass[a4paper,UKenglish,cleveref, autoref, thm-restate]{lipics-v2021}

\pdfoutput=1 
\hideLIPIcs  

\title{A lower bound on the space overhead of fault-tolerant quantum computation} 

\author{Omar Fawzi}{Univ Lyon, ENS Lyon, UCBL, CNRS, Inria, LIP, F-69342, Lyon Cedex 07, France}{omar.fawzi@ens-lyon.fr}{}{European Research Council Project AlgoQIP (grant no. 851716)}

\author{Alexander M\"uller-Hermes}{Institut Camille Jordan, Universit\'{e} Claude Bernard Lyon 1, 69622 Villeurbanne cedex, France \and Department of Mathematics, University of Oslo, P.O. box 1053, Blindern, 0316 Oslo, Norway}{muellerh@math.uio.no}{}{The European Union's Horizon 2020 research and innovation programme under the Marie Sk\l odowska-Curie Action TIPTOP (grant no. 843414)}

\author{Ala Shayeghi}{Univ Lyon, ENS Lyon, UCBL, CNRS, Inria, LIP, F-69342, Lyon Cedex 07, France}{ala.shayeghi89@gmail.com}{}{European Research Council Project AlgoQIP (grant no. 851716)}

\authorrunning{O. Fawzi and A. M\"uller-Hermes and A. Shayeghi} 

\Copyright{Omar Fawzi and Alexander M\"uller-Hermes and Ala Shayeghi} 



\category{} 

\relatedversion{} 

\acknowledgements{We would like to thank Cambyse Rouz\'e and Daniel Stilck Fran\c{c}a for helpful discussions. We also thank Christoph Hirche, Roberto Rubboli and Marco Tomamichel for discussions on contraction coefficients for the quantum relative entropy and their application to the results of this paper.}

\nolinenumbers 

\EventEditors{Mark Braverman}
\EventNoEds{1}
\EventLongTitle{13th Innovations in Theoretical Computer Science Conference (ITCS 2022)}
\EventShortTitle{ITCS 2022}
\EventAcronym{ITCS}
\EventYear{2022}
\EventDate{January 31--February 3, 2022}
\EventLocation{Berkeley, CA, USA}
\EventLogo{}
\SeriesVolume{215}
\ArticleNo{115}
\usepackage{amsmath}
\usepackage{amssymb}

\usepackage{mathtools}
\usepackage{dsfont}
\usepackage{pgfplots}
\pgfplotsset{compat=1.13}
\usetikzlibrary{arrows}

\def\beq{\begin{equation}}
\def\eeq{\end{equation}}
\def\bq{\begin{quote}}
\def\eq{\end{quote}}
\def\ben{\begin{enumerate}}
\def\een{\end{enumerate}}
\def\bit{\begin{itemize}}
\def\eit{\end{itemize}}

\def\lb{\left(}
\def\rb{\right)}

\def\r|{\right|}
\def\lbr{\left[}
\def\rbr{\right]}

\newcommand{\C}{{\mathbb C}}
\newcommand{\R}{{\mathbb R}}
\newcommand{\N}{{\mathbb N}}

\newcommand {\minusspace} {\: \! \!}

\newcommand {\fn} [2] {\ensuremath{ #1 \minusspace \br{ #2 } }}

\newcommand {\br} [1] {\ensuremath{ \left( #1 \right) }}
\newcommand {\Br} [1] {\ensuremath{ \left[ #1 \right] }}
\newcommand {\cbr} [1] {\ensuremath{ \left\lbrace #1 \right\rbrace }}
\newcommand {\norm} [1] {\ensuremath{ \left\| #1 \right\| }}
\newcommand {\normsub} [2] {\ensuremath{ \norm{#1}_{#2} }}

\newcommand {\onenorm} [1] {\normsub{#1}{1}}
\newcommand {\twonorm} [1] {\normsub{#1}{2}}

\newcommand {\bra} [1] {\ensuremath{ \left\langle #1 \right| }}
\newcommand {\ket} [1] {\ensuremath{ \left| #1 \right\rangle }}
\newcommand {\ketbratwo} [2] {\ensuremath{ \left| #1 \middle\rangle \middle\langle #2 \right| }}
\newcommand {\ketbra} [1] {\ketbratwo{#1}{#1}}

\newcommand{\proj}[2]{| #1 \rangle\!\langle #2 |}

\newcommand {\Tr} {\ensuremath{ \mathrm{Tr} }}
\newcommand {\id} {\ensuremath{\mathds{1}}}

\newcommand{\ra} {\rightarrow}
\newcommand{\rX} {\mathrm{X}}
\newcommand{\rY} {\mathrm{Y}}
\newcommand{\rZ} {\mathrm{Z}}

\newcommand{\cB} {\mathcal{B}}
\newcommand{\cC} {\mathcal{C}}
\newcommand{\cD} {\mathcal{D}}
\newcommand{\cE} {\mathcal{E}}
\newcommand{\cI} {\mathcal{I}}

\newcommand{\cN} {\mathcal{N}}
\newcommand{\cM} {\mathcal{M}}
\newcommand{\cS} {\mathcal{S}}
\newcommand{\cT} {\mathcal{T}}

\newcommand{\cL} {\mathcal{L}}
\newcommand{\sE} {\mathscr{E}}

\newcommand{\suppress}[1]{}

\newcommand{\dsep}[1]{d^1_{\mathrm{Sep}\br{#1}}\!}
\newcommand{\chisep}[1]{\chi^2_{\mathrm{Sep}\br{#1}}\!}

\begin{document}

\maketitle

\begin{abstract}
    The threshold theorem is a fundamental result in the theory of fault-tolerant quantum computation stating that arbitrarily long quantum computations can be performed with a polylogarithmic overhead provided the noise level is below a constant level. A recent work by Fawzi, Grospellier and Leverrier (FOCS 2018) building on a result by Gottesman (QIC 2013) has shown that the space overhead can be asymptotically reduced to a constant independent of the circuit provided we only consider circuits with a length bounded by a polynomial in the width. In this work, using a minimal model for quantum fault tolerance, we establish a general lower bound on the space overhead required to achieve fault tolerance. 
    
    For any non-unitary qubit channel $\cN$ and any quantum fault tolerance schemes against $\mathrm{i.i.d.}$ noise modeled by $\cN$, we prove a lower bound of $\max\cbr{\mathrm{Q}(\cN)^{-1}n,\alpha_\cN \log T}$ on the number of physical qubits, for circuits of length $T$ and width $n$. Here, $\mathrm{Q}(\cN)$ denotes the quantum capacity of $\cN$ and $\alpha_\cN>0$ is a constant only depending on the channel $\cN$. In our model, we allow for qubits to be replaced by fresh ones during the execution of the circuit and we allow classical computation to be free and perfect. This improves upon results that assumed classical computations to be also affected by noise, and that sometimes did not allow for fresh qubits to be added. Along the way, we prove an exponential upper bound on the maximal length of fault-tolerant quantum computation with amplitude damping noise resolving a conjecture by Ben-Or, Gottesman and Hassidim (2013). \footnote{An earlier version of this paper appeared in proceedings of ITCS 2022. In the current version, we have extended our results to the model with noiseless classical computation for all qubit noise channels.}
\end{abstract}

\section{Introduction}
Quantum computing is capable of solving problems for which no efficient classical algorithms are known. However, the decoherence of quantum systems due to inevitable interactions with the environment makes it challenging to build reliable large-scale quantum computers. To circumvent this obstacle the theory of quantum error correction and fault tolerance was developed~\cite{Shor95,Shor96,Stean96}.

A quantum fault tolerance scheme receives as input any ideal circuit designed for noiseless computation and outputs an encoding of the circuit which is robust against noise. The fault-tolerant threshold theorem was a major achievement in quantum computing: As long as the noise level in a quantum computer is below a certain threshold, any arbitrarily long quantum computation can be reliably performed with a reasonably low overhead, namely, a polylogarithmic factor in space and time~\cite{AB97,KLZ98,Kit97,AGP05}. 
Since then a lot of effort has been focused on introducing fault tolerance schemes with lower resource overheads, improving the lower bounds on the fault-tolerant threshold for different noise models (see, e.g., \cite{Knill05,Aliferis08,Aliferis09,Wang_2011,Yao12}). Gottesman~\cite{gottesman13constant} showed that if quantum error correcting codes with certain properties exist, then the space overhead can be brought down to a constant arbitrarily close to $1$, for computations in which the number of time steps is at most polynomial in the number of qubits. Recently, Fawzi \emph{et al.\/}~\cite{fawzi18constant} showed that these properties are indeed satisfied by quantum expander codes and the statement is true even for sub-exponential computations. But what is the fundamental limit on the space overhead for quantum fault-tolerance as a function of the noise model? This question still remains open. More broadly, the study of general properties and limitations of quantum fault tolerance schemes has received far less attention. This is in part due to the difficulty of providing a definition encompassing the different aspects of quantum fault tolerance. In this paper, we take a step towards proving general statements about quantum fault tolerance schemes. In particular, given an \emph{arbitrary} single qubit noise channel $\cN$, our main result is a lower bound on the achievable space overhead of quantum fault tolerance schemes against $\mathrm{i.i.d.}$ errors modeled by $\cN$.

Previous works on the limitations of noisy quantum circuits typically proved statements of the following form (see, e.g., \cite{Razborov,BenOr,Buhrman,Kempe}: If the noise level is too high, then any noisy circuit is ``useless'' after a limited number of time steps. These results do not apply to the low-noise regime, i.e., when the noise channel $\cN$ is not far from the identity channel, and in particular, noise levels below the fault-tolerant threshold. Therefore, they cannot be used to derive lower bounds on the overhead of fault tolerance schemes. Moreover, these statements crucially depend on several parameters, including the circuit model, the noise model, and the measure of uselessness they consider. Tweaking any of these variables can make a substantial difference. The circuit model specifies the basic components of the quantum circuits including the set of possible gates, state preparations, and measurements, as well as a model of classical computation for the processing of the measurement outcomes by the circuits. 
The noise model specifies how each component of the circuit deviates from its ideal action due the noise. 
Different noise models such as depolarizing, dephasing, or amplitude damping noise are not equivalent in the context of fault tolerance (see, e.g.,~\cite{BenOr}). Earlier fault tolerance schemes were designed for generic noise models. In practice, however, any physical implementation of a quantum computer has some additional noise structure. With the knowledge of the dominant noise in the system, one can potentially utilize schemes that are tailored for a specific noise model for cheaper suppression of the dominant noise. Indeed, biased-noise fault tolerance beyond the usual Pauli noise models has gained growing interest in recent years (see, e.g.,~\cite{AP08,SKR09,BP13,TBF18,catqubits19}). Earlier relevant results on limitations of noisy quantum computation, however, are mostly focused on a particular noise channel, typically, depolarizing noise (with the exception of Ref.~\cite{Virmani}), and it is not clear if the introduced techniques can be extended to other noise models.

With current technology, classical computation can be performed with practically perfect accuracy. Therefore, it is reasonable and conventional to assume classical computation to be error-free when evaluating the performance of quantum fault tolerance schemes. For instance, in this context, given measured syndromes of a quantum stabilizer code, computing a description of a corresponding error is assumed to be noiseless. While some earlier impossibility results establish very strong limitations for noisy circuits (see the discussion in Section~\ref{sec:related}), they do not incorporate the aid of noiseless classical computation and the adaptivity achieved through classical control. In these works, the classical systems are not distinguished from the quantum systems and they are assumed to be subject to the same noise. In fact, without this assumption, the ``uselessness'' statements established in these works fail to hold even for arbitrarily long circuits. Therefore, these results are not directly applicable to the fault tolerance schemes which rely on perfect classical computation and control. By allowing this additional resource in our model, we address a limitation of the previous results on the power of noisy quantum circuits. Our results apply to adaptive protocols composed of a hybrid of classical and quantum computation.

Before stating our results more formally, we first need to describe our quantum circuit model.

\subsection{Our model}

\begin{figure}
\centering
\subfloat[A noiseless quantum circuit. The classical subsystems are denoted by double lines.]{\includegraphics{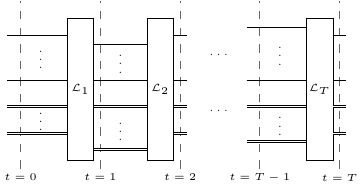}}
\qquad
\subfloat[A noisy circuit in the $\mathrm{i.i.d.}$ noise model: The classical systems are noise-free.]{\includegraphics{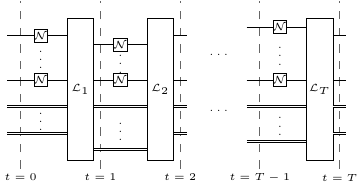}}
\caption{Illustration of the circuit model}
\label{Fig:CircuitModel}
\end{figure}

A (noiseless) quantum circuit $C$ in our model is defined by a sequence $\cbr{\cL_i}_{i\in[T]}$ of arbitrary completely positive and trace-preserving (CPTP) linear maps, acting on a potentially non-trivial input system. For any density operator $\rho$ on the input Hilbert space, the output of the circuit is given by
\begin{equation*}
    C(\rho)\coloneqq \br{\cL_T\circ \cdots \circ \cL_{2} \circ \cL_1} \br{\rho}\enspace,
\end{equation*}
where we use the same notation $C$ to denote the overall channel corresponding to the action of the circuit $C$. 
We use the term \emph{time} to refer to the vertical cuts between the layers of computation, as shown in Fig.~\ref{Fig:CircuitModel}. The $i$-th \emph{time step} of the circuit is defined by $\cL_i$ acting between $t=i-1$ and $t=i$. 

In our model, we distinguish between classical and quantum subsystems and we allow free and noiseless classical computation. Given a qubit noise channel $\cN$, for any quantum circuit $C$, the noisy implementation of $C$ with $\mathrm{i.i.d.}$ noise modeled by $\cN$ is obtained by interlacing the computational steps $\cL_i$ with layer of the noise channel $\cN$ acting independently on every qubit, while as shown in Fig.~\ref{Fig:CircuitModel}, the classical subsystems are unaffected by the noise. \suppress{Let $n_i$ denote the number of qubits at $t=i$. Then for every input state $\rho$, the output of the noisy circuit in this model is given by
\begin{equation*}
    \br{\mathrm{i.i.d._\cN}(C)}\br{\rho}= \br{\cL_T\circ \cN^{\otimes n_{T-1}}\circ\cdots \circ \cL_{2} \circ \cN^{\otimes n_1}\circ \cL_1\circ \cN^{\otimes n_0}} \br{\rho}\enspace,
\end{equation*}}
We denote by $\mathrm{i.i.d._\cN}(C)$ the overall channel corresponding to the noisy implementation of $C$. For a circuit $C$, the \emph{width} $Width(C)$ of the circuit is the maximum number of \emph{qubits} at any time step and its \emph{length} $Length(C)$ is the number of time steps of the circuit. Note that in our model, we do not count the size of the classical subsystems in $Width(C)$; in other words, the width of the circuit is determined by the number of qubit systems that are subject to the noise. Moreover, classical computations do not contribute to $Length(C)$ in this model.

\subsection{Discussion of our main results, proof techniques, and related work}

\textbf{Limitations on quantum memory times.}\suppress{We consider quantum circuits in which the gates can be arbitrary quantum channels, i.e., completely positive and trace-preserving maps, possibly controlled by a classical register of arbitrary (finite) dimension. Moreover, we allow perfect classical computation with unlimited power and perfect classical control.} We consider noisy quantum circuits in the $\mathrm{i.i.d.}_\cN$ error model, for an arbitrary non-unitary qubit channel $\cN$. Our main technical contribution is the following: 

\begin{theorem} \label{thm:no_long_memory}
    Let $\cI$ denote the qubit identity channel and $\cN$ be a non-unitary qubit channel. Then there exists a constant $p\in(0,1]$ only depending on $\cN$ such that the following holds:
    Let $C$ be an arbitrary circuit of length $T$ and width $n$ such that $T\geq \br{2/p}^{2n}$. Then for any pair of ideal encoding and decoding maps $\cE$ and $\cD$ (of appropriate input and output dimensions), we have
    \begin{equation*}
        \onenorm{\cD \circ \mathrm{i.i.d.}_\cN\br{C} \circ \cE - \cI} \geq \epsilon_0 \enspace,
    \end{equation*}
    for some universal constant $\epsilon_0\geq 1/128$. 
\end{theorem}

The constant $p=p(\cN)$ in Theorem \ref{thm:no_long_memory} is given by $p(\cN)=\max(p_1(\cN),p_2(\cN))$ where
\[
p_1(\cN) = \max\left\lbrace q\in \left[ 0,1\right]:\cN\geq q\cM ,\text{ for an entanglement breaking qubit channel }\cM \right\rbrace.
\]
and
\begin{align*}
p_2(\mathcal{N}) = \max\left\lbrace q^2\lambda_{\min}\left(C_{\mathcal{M}^\dagger \circ \mathcal{M}}\right)/16: \right.&q\in \left[ 0,1\right],~\mathcal{N}\geq q\mathcal{M}, \\
&\left.\text{for a non-unital extremal qubit channel }\mathcal{M}\right\rbrace.
\end{align*}
Here, we write $\cN\geq \cM$ for linear maps $\cN$ and $\cM$ if the difference $\cN-\cM$ is completely positive. We will show that $p_1(\cN)>0$ for any unital qubit channel $\cN$ that is not unitary, and that $p_2(\cN)>0$ for every non-unital qubit channel $\cN$. 

We point out that the statement of Theorem \ref{thm:no_long_memory} does not hold as is for quantum circuits operating on qudits rather than qubits and affected by any non-unitary noise channel. By taking a direct sum of the ideal quantum channel and a completely depolarizing channel it is easy to construct non-unitary noise channels for which quantum information can be stored indefinitely in a decoherence-free subspace. To obtain an upper bound on the storage time in these cases, further assumptions are needed.

Theorem~\ref{thm:no_long_memory} states that for any noise channel $\cN$, no noisy quantum circuit in the $\mathrm{i.i.d.}_\cN$ error model is capable of preserving even a single (unknown) qubit for more than an exponential number of time steps in its width. We would like to emphasize two features of our result:
\begin{enumerate}
\item The quantum channels making up the circuit do not need to be unital. Therefore, they can remove all the entropy produced by the noise channel along the computation, for example by introducing ancilla qubits.
\item Our results also hold for circuits with free and noiseless classical computation. Even after long computation times, it is then possible to prepare orthogonal states conditioned on earlier measurement outcomes that have been kept in the perfect classical memory.
\end{enumerate}
In order to prove the result, we consider two parallel instances of the circuit $C$ and prove that for any input state the distance from the set of separable states with respect to this bipartition is contracted by a factor of $1-p^{2n}$ (with $p>0$ as above) in each time step. This implies that the circuit $C$ itself cannot preserve every quantum state for more than $\br{2/p}^{2n}$ time steps.

Ben-Or \emph{et al.\/} \cite{BenOr} studied the limitations of fault-tolerant quantum computation under the independent noise model $\mathrm{i.i.d.}_\cN$, for different classes of non-unitary qubit channels. Applying such a channel many times either converges to a limit channel mapping every quantum state to a point in the Bloch sphere (e.g., for the depolarizing channel), or it converges to a limit quantum channel mapping every quantum state to a diameter of the Bloch sphere (e.g., for the dephasing channel). By slightly generalizing an earlier entropic argument\footnote{This argument was stated in \cite{ABIN96}, but the proof of a key lemma is not correct. This was corrected in~\cite{muller2016relative}.}, Ben-Or \emph{et al.\/} prove that if the limit channel maps every quantum state to the maximally mixed state, then it is impossible to compute for more than $\mathrm{O}(\log n)$ steps and they show how to achieve this bound up to a polylogarithmic factor in $\log n$. Entropic arguments of this kind are crucially based on the assumption that entropy cannot be removed during the execution of the circuit. In particular, these arguments do not apply for our circuit model (including common settings of fault-tolerance~\cite{AB97,AGP05}) where the entropy produced by the noise channel can be removed along the computation (cf., feature 1 from above). When the limiting quantum channel maps every quantum state to a diameter of the Bloch sphere (e.g., the dephasing channel), Ben-Or \emph{et al.\/} prove that it is impossible to store quantum information for more than a polynomial number of time steps, however again assuming no fresh ancilla qubits and no noiseless classical computation. Using standard threshold theorems, they also show that for this class of channels, computations involving $\mathrm{O}(n^a)$ qubits for $\mathrm{O}(n^b)$ time steps are possible, provided $a+b<1$. Finally, if the limit channel has a fixed point other than the center of the Bloch sphere (e.g., the amplitude damping channel) then exponentially long quantum computations are possible. This is achieved by using the noise itself as a refrigerator to cool down the qubits and produce fresh ancillas. They also conjecture that for this class of channels it is impossible to compute for more than exponential time with no fresh ancilla qubits. To our knowledge, prior to this work, no proof of this statement existed. Theorem~\ref{thm:no_long_memory} proves this conjecture to be true. In addition, our result applies even if fresh ancillas are used and the upper bound on the length becomes $e^{\mathrm{O}(n+m)}$ with $m$ ancilla qubits in each time step.

In a related work, Raginsky~\cite{Raginsky} considers the noise model where a fixed noise channel $\cT$ acts collectively on all the qubits in the quantum circuit after each computational step. He shows that, for any quantum circuit involving only unitary operations, if the channel $\cT$ modeling the noise is strictly contractive then the distance of possible output states of the circuit is exponentially small in the length of the circuit. Moreover, he shows that the set of strictly contractive channels is dense in the set of all quantum channels. Using this fact and under the assumption that quantum operations can be performed only with a finite precision, he argues that strictly contractive channels serve as a natural abstract model for noise in any physically realizable quantum computer. Our work can be seen as an extension of this argument. In particular, we show how to apply a similar reasoning when the channel $\cT$ modeling the noise is given by tensor powers of an arbitrary non-unitary qubit channel, and we further extend the argument by taking into account free and perfect classical computation (obtaining feature 2 from above). Note that in general strict contractivity of a quantum channel does not necessarily imply the same property for its tensor powers (see the discussion at the end of Section~\ref{sec:exp bnd} for more details).  

\textbf{Limitations on quantum fault-tolerance.} Inspired by the observation of Theorem~\ref{thm:no_long_memory}, in Section~\ref{sec:implications}, we provide a natural high-level description of quantum fault tolerance schemes with minimal assumptions. A fault tolerance scheme is a circuit encoding associated with a family of quantum error correcting codes. Intuitively, we only require a fault tolerance scheme to be capable of implementing any arbitrary logical circuit on \emph{encoded data} in a code from the family, up to arbitrary accuracy in the presence of noise. Here, we assume that the encoding and decoding maps for the quantum error correcting code can be implemented perfectly. We quantify the accuracy in terms of the induced trace distance between the channels corresponding to the ideal circuit and its fault-tolerant implementation with noise. To our knowledge, this requirement is satisfied by most, if not all, known fault tolerance schemes. 

In this work we do not focus on investigating the limitations of \emph{computation} using noisy quantum circuits. Indeed, our results are shown assuming perfect classical computation with \emph{unlimited} power is possible. We are rather interested in proving limitations of fault tolerance schemes as procedures which, \emph{oblivious} to the overall computation being performed, allow the implementation of any ideal input circuit using noisy components.

Our high-level definition of fault tolerance enables us to recast the problem of proving a lower bound on the space overhead for fault-tolerant quantum computation into a more tractable problem in the better understood framework of quantum Shannon theory. Note that a fault tolerance scheme should naturally be able to implement a quantum memory. From this alternative perspective, a quantum memory for $T$ time steps in the $\mathrm{i.i.d.}_\cN$ noise model corresponds to one-way point-to-point quantum communication with $T-1$ intermediate relay points on a line, linked by $T$ instances of the noisy channel $\cN$, where the relay points are allowed to perform arbitrary (noiseless) quantum operations and communicate classically with the following relay points on the line. Similar communication models have been studied earlier in the quantum Shannon theory literature (see, e.g.,~\cite{Pirandola19}). Here, given a qubit channel $\cN$, we are interested in the optimal rate of communication over the links for reliably sending $n$ qubits as a function of $T$ and $n$. The case $T=1$ is well-studied in the quantum Shannon theory literature and the optimal rate is given by the quantum capacity of $\cN$. Theorem \ref{thm:no_long_memory} implies an upper bound on the achievable rates (or, equivalently, a lower bound on the necessary overhead) as a function of $T$. The following is our main result:

\begin{theorem}\label{thm:overhead}
    Let $\cN$ be a non-unitary qubit channel. For any fault tolerance scheme against the $\mathrm{i.i.d.}_\cN$ noise model, the number of physical qubits is at least 
    \begin{equation*}
        \max\cbr{\mathrm{Q}(\cN)^{-1}n,\alpha_\cN \log T}\enspace,
    \end{equation*}
    for circuits of length $T$ and width $n$, and any constant accuracy $\epsilon\leq 1/128$. Here, $\mathrm{Q}(\cN)$ denotes the quantum capacity of $\cN$ and $\alpha_\cN=\frac{1}{2\log 2/p}$, where $p\in(0,1]$ is a constant only depending on the channel $\cN$. In particular, fault tolerance is not possible even for a single time step if $\cN$ has zero quantum capacity or the space overhead is strictly less than $\mathrm{Q}(\cN)^{-1}$. 
\end{theorem}

Theorem~\ref{thm:overhead} in particular implies that fault-tolerant quantum computation with constant space overhead is not achievable beyond exponential computations for any noise channel $\cN$. 

For the $p$-depolarizing noise given by $\cN(\rho) = (1-p) \rho + p \frac{\id}{2}$, the quantum capacity is known to be zero for $p > \frac{1}{3}$~\cite{BDE98}. Thus, for $p > \frac{1}{3}$, the lower bound becomes $+\infty$, i.e., fault tolerance is not possible in our model. This gives stronger limitations compared to previous results~\cite{Razborov,Kempe,Buhrman} (except for a result in~\cite{Kempe} where making an assumption on the gate set, namely CNOT is the only allowed two-qubit gate, they show fault-tolerance is not possible when $p \geq 29.3\%$). We stress that all these results are incomparable as the models are different. For $p\leq 1/3$, we cannot rule out fault tolerance, but we can show a lower bound on the overhead as a function of the circuit depth.

\subsection{Related work on limitations of quantum circuits under strong noise} \label{sec:related}

 The results discussed below are not directly relevant to our work in the sense that they do not apply to noise levels below the fault-tolerant threshold. However, they establish impossibility results for noisy circuits and prove upper bounds on the threshold for fault-tolerant computation.

Razborov \cite{Razborov} considers $p$-depolarizing $\mathrm{i.i.d.}$ noise model and circuits with quantum gates which are allowed to be arbitrary quantum channels. He proves that if $p>1-1/k$ then for any circuit in this model of width $n$ and length $T$ with gates of fan-in at most $k$, the distance between the states obtained by applying the circuit on any pair of input states decreases as $n2^{-\Omega(T)}$. In particular, any computation of length $T=\Omega(\log n)$ is essentially useless since the output is independent of the initial state. Such a strong notion of uselessness and upper bound on $T$ as a function of $n$ is possible in this case because of the strong assumptions: the noise is sufficiently high and classical information is subject to noise. 
If classical information can be stored perfectly, one can imagine a circuit in which part of the input is measured in the first step, stored, and then output at the end (as in the second strategy given above): for arbitrarily large values of $T$ this circuit is not useless. Moreover, in our circuit model, there is no restriction on the gate set of the circuit, in particular the fan-in of the gates can be arbitrary. Our bound applies to the $\mathrm{i.i.d.}$ noise model for any non-unitary qubit noise channel. Finally, we point out that for the model considered by Razborov, our techniques can be adapted to recover similar bounds for a large class of quantum channels. 

Using a different approach, Kempe \emph{et al.\/} \cite{Kempe} improve the threshold upper bound to $1-\Theta(1/\sqrt{k})$. However, their result is not quite comparable with Razborov's earlier result since they are obtained in different models. Kempe \emph{et al.\/} consider circuits composed of arbitrary one qubit CPTP gates which are assumed to be essentially noiseless and $k$-qubit gates that are probabilistic mixtures of unitary operations preceded by independent $p$-depolarizing noise on their input qubits. They also assume that the output is the outcome of a measurement of a designated qubit (or a constant number of qubits) in the computational basis. They prove that if $p>1-\Theta(1/\sqrt{k})$ then for any circuit of length $T$ in this model, the total variation distance between the output distributions for any pair of input states decreases as $2^{-\Omega(T)}$. Similar to Razborov's, their model does not allow noiseless classical computation and classically controlled operations. Note that a similar bound is known to hold for noisy classical circuits as well \cite{Evans1,Evans2}.

When classical computation is assumed to be free and perfect, a measure of ``uselessness'' for noisy quantum circuits should naturally capture quantum phenomena that are not observed in classical circuits. One such measure that has been studied in earlier works \cite{Virmani} is being efficiently simulable classically. Along this direction, Buhrman \emph{et al.\/} \cite{Buhrman} show that quantum circuits consisting of noiseless stabilizer operations and arbitrary single qubit unitary gates followed by $45.3\%$ depolarizing noise can be efficiently simulated by classical circuits. Their result is incomparable to ours for the same reasons above.

\section{Preliminaries} \label{sec:prelim}

Let $\cM_d$ denote the set of $d\times d$ matrices with complex entries. Consider a quantum system described by the $d$-dimensional Hilbert space $\C^d$. Quantum states of the system are described by the set of \emph{density operators} given by $\cbr{\rho\in\cM_d:\rho\geq0,\Tr(\rho)=1}$. Pure state density operators are rank-one projectors $\rho=\ketbra{\psi}$. A density operator $\zeta$ on $\C^{d_A}\otimes\C^{d_B}$ is called \emph{separable} if it can be written as $\zeta = \sum_{i\in I} p_i\, \rho_i \otimes \sigma_i \enspace$ for some finite set $I$, a probability distribution $p:I\rightarrow [0,1]$, and density operators $\{\rho_i:i\in I\}$ on $\C^{d_A}$ and $\{\sigma_i:i\in I\}$ on $\C^{d_B}$. We denote by $\mathrm{Sep(A:B)}$ the set of all separable state on $\C^{d_A}\otimes\C^{d_B}$ with respect to the bi-partition $A:B$, where the two subsystems are labeled $A$ and $B$. A quantum state is called \emph{entangled} if it is not separable.

\textbf{Representation of Quantum channels.} A \emph{quantum channel} is a completely positive and trace-preserving linear map $\cT:\cM_{d_A}\rightarrow\cM_{d_B}$. We denote by $\cI_d: \cM_d \rightarrow \cM_d$ the identity channel, i.e., the identity map on $\cM_d$. Next, we review some basics about different representations of quantum channels (for more details see~\cite{TQI_Watrous}). Let $\cT:\cM_{d_A}\rightarrow\cM_{d_B}$ be a quantum channel. Let
\begin{equation}\label{eq:def_ChoiMatrix}
    C_\cT=\textstyle\sum_{i,j\in[d]} \cT\br{\ketbratwo{i}{j}}\otimes \ketbratwo{i}{j} \enspace,
\end{equation}
where $\cbr{\ket{i}:i\in [d_A]}$ is the standard basis of $\C^{d_A}$. The operator $C_\cT$ is called the \emph{Choi matrix} of $\cT$. The rank of $C_\cT$ is called the Choi rank of $\cT$.
For any quantum channel $\cT$, there exists a finite set $I$ and linear operators $K_i:\C^{d_A}\rightarrow \C^{d_B}$ for any $i\in I$ such that $\sum_{i\in I} K_i^\dagger K_i=\id_{d_A}$ and $\cT(X)=\sum_{i\in I} K_i X K_i^\dagger$, for all $X \in \cM_{d_A}$. Any representation of this form is called a \emph{Kraus representation} of the channel $\cT$ and the operators $\cbr{K_i}_{i\in I}$ are called the corresponding \emph{Kraus operators}. Note that Kraus representation in general is not unique, and the minimum number of Kraus operators in any Kraus representation of $\cT$ is equal to the Choi rank of $\cT$~\cite[Corollary 2.23]{TQI_Watrous}. Important classes of quantum channels can be defined via their Kraus representation: A quantum channel is called \emph{unitary} if it admits a Kraus decomposition with a single unitary Kraus operator, and it is called \emph{entanglement breaking} if it admits a Kraus representation with Kraus operators $K_i$ of rank $1$. Consider quantum systems labelled $A_1,A_2,B_1$ and $B_2$. A bipartite quantum channel $\cT:\cM_{d_{A_1}}\otimes \cM_{d_{B_1}}\ra \cM_{d_{A_2}}\otimes \cM_{d_{B_2}}$ is called a \emph{separable quantum channel} with respect to the bipartition $(A_1,A_2):(B_1,B_2)$ if it admits a Kraus representation with Kraus operators of the form $K_i = K^{A}_i\otimes K^{B}_i$ with operators $K^{A}_i:\C^{d_{A_1}}\ra \C^{d_{A_2}}$ and $K^{B}_i:\C^{d_{B_1}}\ra \C^{d_{B_2}}$. We denote the set of these quantum channels by $\mathrm{SepC}\br{(A_1,A_2):(B_1,B_2)}$. The set of quantum channels from $\cM_{d_A}$ to $\cM_{d_B}$ is compact and convex, and its extreme points are given by the quantum channels $\cT:\cM_{d_A}\rightarrow\cM_{d_B}$ having a Kraus representation $\cT\br{X}=\sum_i K_i X K_i^{\dagger}$ such that $\cbr{K_i^\dagger K_j}_{i,j}$ is a linearly independent set~\cite{Choi}. For any quantum channel $\cT:\cM_{d_A}\rightarrow\cM_{d_B}$, there exists an isometry $V:\C^{d_A} \rightarrow \C^{d_B}\otimes \C^{d_E}$, for some $d_E\in \N$, such that $\cT(X)=\Tr_E\br{VXV^\dagger}$, for every $X\in\C^{d_A}$. Any such representation of $\cT$ is called a \emph{Stinespring representation} and the isometry $V$ is referred to as a \emph{Stinespring dilation} of $\cT$. The dimension $d_E$ can be chosen to be equal to the Choi rank of $\cT$~\cite[Corollary 2.27]{TQI_Watrous}. 

\textbf{Trace-distance, quantum $\chi^2$-divergence, and contraction coefficients.} The \emph{induced trace norm} of a linear map $\cT:\cM_{d} \rightarrow \cM_{d'}$ is defined as 
\begin{equation*}
    \onenorm{\cT} \coloneqq \max \cbr{\onenorm{\cT(X)}\,:\, X\in \cM_d,\, \onenorm{X}\leq 1} \enspace.
\end{equation*} 
We will need the \emph{trace-norm contraction coefficient} of a quantum channel $\cT:\cM_d\rightarrow\cM_d$ given by 
\begin{equation}\label{equ:QDobrushin}
    \eta_{\text{tr}}\br{\cT} := \sup_{\rho,\sigma} \frac{\onenorm{\cT(\rho)-\cT(\sigma)}}{\onenorm{\rho-\sigma}
    } \enspace,    
\end{equation}
where the supremum is over all pairs of quantum states. This quantity can be regarded as a quantum analogue of the Dobrushin ergodicity coefficient \cite{Dobrushin}, which is an important tool in the study of Markov processes. We will often use that 
\begin{equation}\label{equ:TNContract}
    \eta_{\text{tr}}\br{\cT} = \frac{1}{2}\max_{\psi\perp\phi}\onenorm{ \cT\br{\proj{\psi}{\psi}-\proj{\phi}{\phi}}} \enspace,
\end{equation}
where the maximum is over orthogonal pairs of pure states (see~\cite{ruskai1994beyond}). In particular, $\eta_{\text{tr}}(\cT)=1$ if and only if, there exists a pure state $\psi\in\C^d$ for which $(\cT^\dagger \circ \cT) \br{\proj{\psi}{\psi}}$ is not full-rank.

For a pair of density operators $\rho$ and $\sigma$, the $\chi^2$-divergence is defined by 
\begin{align*}
\chi^2(\rho, \sigma) &\coloneqq \Tr\left((\rho - \sigma) \sigma^{-1/2} (\rho - \sigma) \sigma^{-1/2}) \right) \ ,
\end{align*}
when $\mathrm{supp}(\rho) \subseteq \mathrm{supp}(\sigma)$ and $+\infty$, otherwise. The operator $\sigma^{-1/2}$ is to be understood as generalized inverse of $\sigma^{1/2}$. Note that we have $\chi^2(\rho, \sigma) = \Tr(\rho \sigma^{-1/2} \rho \sigma^{-1/2}) - 1$. 

$\chi^2$-divergence is monotone under quantum channels and jointly convex, see e.g.,~\cite[Proposition 7]{TKRWV10}. Moreover, we have 
\begin{equation} \label{eq:tr_dist vs chi2}
    \onenorm{\rho-\sigma}^2\leq \chi^2(\rho,\sigma)\enspace,
\end{equation}
for any pair of density operators $\rho$ and $\sigma$~\cite[Lemma 5]{TKRWV10}. We will use the \emph{$\chi^2$ contraction coefficient} of a quantum channel $\cT:\cM_d\rightarrow\cM_d$ given by
\begin{align}
\eta_{\chi}\br{\cT} &\coloneqq \sup_{\rho, \sigma} \frac{\chi^2(\cT(\rho), \cT(\sigma))}{\chi^2(\rho, \sigma)}\enspace,
\end{align}
where the supremum is over all pairs of quantum states. For any quantum channel $\cT:\cM_d\rightarrow\cM_d$, the following inequality holds~\cite[Theorem 14]{TKRWV10}:
\begin{equation} \label{eq:eta_chi vs eta_tr}
    \eta_{\chi}\br{\cT} \leq \eta_{\mathrm{tr}}\br{\cT}\enspace.
\end{equation}

\textbf{Quantum capacity.} The quantum capacity of a quantum channel is the maximum rate at which it can reliably transmit quantum information over asymptotically many uses of the channel.

\begin{definition} \label{def:Qcapacity}
    For a quantum channel $\cT:\cM_{d_A}\rightarrow\cM_{d_B}$, a communication rate $R$ is called $\epsilon$-\emph{achievable} if there exists $n_\epsilon$ such that for all $n\geq n_\epsilon$, there is an encoding quantum channel $\cE_n:\cM^{\otimes Rn}_{2}\rightarrow \cM^{\otimes n}_{d_A}$ and a decoding channel $\cD_n:\cM^{\otimes n}_{d_B}\rightarrow \cM^{\otimes Rn}_{2}$ such that $\onenorm{(\cD_n\circ \cT^{\otimes n}\circ \cE_n)(\rho)-\rho}\leq \epsilon$, for all $\rho\in\cM^{\otimes Rn}_2$. A rate $R$ is an \emph{achievable communication rate} if it is $\epsilon$-achievable for all $\epsilon\in(0,2]$. The \emph{quantum capacity} of $\cT$ is the supremum over all achievable rates.

\end{definition}

\textbf{Representation of qubit channels.} Consider the set of single-qubit Pauli operators $\mathcal{P}=\{\id,\rX,\rY,\rZ\}$ defined as
\begin{equation}
    \id=\begin{pmatrix} 1&0\\ 0&1 \end{pmatrix},\;
    \rX=\begin{pmatrix} 0&1\\ 1&0 \end{pmatrix},\;
    \rY=\begin{pmatrix} 0&-\mathrm{i}\\ \mathrm{i}&0 \end{pmatrix},\;
    \rZ=\begin{pmatrix} 1&0\\ 0&-1 \end{pmatrix}.
\end{equation}
The set of $2\times 2$ Hermitian matrices is a real vector space. One can easily verify that $\mathcal{P}$ forms an orthogonal basis for this set, with respect to the Hilbert-Schmidt inner product defined as $\langle A,B \rangle=\mathrm{Tr}\br{A^\dagger B}$. In particular, any single-qubit quantum state $\rho$ can be written in the \emph{Bloch sphere representation} as 
\begin{equation}
    \rho = \frac{1}{2} \br{\id+r\cdot \sigma}\enspace,
\end{equation}
where $r=\br{r_x,r_y,r_z}\in \R^3$ with $\twonorm{r} \leq 1$ and $\sigma$ is the vector of Pauli matrices $\br{\rX,\rY,\rZ}$. It has been shown in~\cite{Ruskai02}, that the action of any qubit channel $\cN:\cM_2\ra \cM_2$ can be written as $\cN\br{\rho}=U\fn{\cC_{[t,\lambda]}}{V\rho V^\dagger}U^\dagger$, with 
\begin{equation}
    \fn{\cC_{[t,\lambda]}}{\frac{1}{2}\Br{\id+r\cdot \sigma}} = \frac{1}{2}\br{\id+\br{t+\Lambda r}\cdot \sigma} \enspace,
\end{equation}
where $t=\br{t_x,t_y,t_z}\in \R^3$ with $\twonorm{t} \leq 1$ and $\Lambda=\mathrm{diag}\br{\lambda}$ with $\lambda=\br{\lambda_x,\lambda_y,\lambda_z}\in \R^3$ and $\norm{\Lambda}\leq 1$. The qubit channel $\cN$ is \emph{unital}, i.e., $\cN(\id)=\id$, if and only if $t=0$ in this representation. Imposing complete positivity of the unital map $\cC_{[0,\lambda]}$ is equivalent to restricting the parameters $\lambda\in\R^3$ to a regular tetrahedron~\cite{Ruskai02} with corners corresponding to the conjugation by the Pauli operators $\id,\rX,\rY,\rZ$, respectively. Specifically, the set of unital qubit channels is given by
\begin{equation}\label{eq:tetrahedron}
        \cbr{U\fn{\cC_{[0,\lambda]}}{V\rho V^\dagger}U^\dagger : U,V \mathrm{ unitary}, \lambda \in \mathrm{conv}\!\cbr{\br{1,1,1},\br{1,-1,-1},\br{-1,1,-1},\br{-1,-1,1}} }.
\end{equation}
The middle points on the edges of the tetrahedron of Eq.~\eqref{eq:tetrahedron} are the permutations of $\br{\pm1,0,0}$, and they are the vertices of a regular octahedron corresponding to unital entanglement breaking qubit channels. From these geometric considerations and the fact that entanglement breaking channels stay entanglement breaking under composition by unitary quantum channels, we obtain the following lemma:

\begin{lemma}\label{lem:UnitalCaseLemma}
Any unital qubit channel $\cN:\cM_2\ra \cM_2$ can be written as a convex combination
\[
\cN = (1-p)\cT + p\cB,
\]
where $\cT:\cM_2\ra\cM_2$ is a unitary quantum channel and $\cB:\cM_2\ra\cM_2$ is entanglement breaking. 
\end{lemma}

\section{Maximal length of noisy quantum computation} \label{sec:exp bnd}

Let $C$ be a quantum circuit of width $n$ and length $T$ defined by a sequence of quantum channels $\cbr{\cL_i}_{i\in[T]}$, as depicted in Fig.~\ref{fig:figure2}. For $i\in\{0,\ldots,T\}$, let $(A_i:B_i)$ be a bipartition of the registers of the circuit after $i$ time steps into disjoint subsets such that $\cL_i\in \mathrm{SepC}\br{A_{i-1},A_{i}:B_{i-1},B_i}$, i.e., $\cL_i$ in each time step is a separable channel with respect to the bipartition of the subsystems into the $A$ and $B$ parts. Furthermore, for every $i\in\cbr{0,\ldots,T}$, let $A_i=A^c_iA^q_i$ and $B_i=B^c_iB^q_i$, where $A^c_i$ and $B^c_i$ are classical subsystems of $A_i$ and $B_i$, respectively, of arbitrary finite dimensions and $A^q_i$ and $B^q_i$ are the quantum subsystems which contain a total of $n_i$ qubits. Let $\cN$ be an arbitrary non-unitary qubit channel and consider the noisy implementation of $C$ in the $\mathrm{i.i.d._\cN}$ model. Recall that in our model only the quantum subsystems are subject to the noise channel $\cN$. We prove that $\mathrm{i.i.d.}_\cN(C)$ is not capable of preserving the entanglement between the $A$ and $B$ subsystems for more than exponential time in the width of $C$.

We define
\begin{align*}
\chisep{A:B}\br{\tau_{AB}} &= \min_{\sigma \in \mathrm{Sep}\br{A:B}} \chi^2(\tau_{AB}, \sigma_{AB}) \ .
\end{align*}
It is simple to see that this measure is monotone under separable quantum channels and convex in the state.

We denote by $\rho_{AB}^i$ the state of the noisy circuit after $i$ time steps. The following lemma is our main technical contribution which is of independent interest.

\begin{figure}
    \centering
    \includegraphics{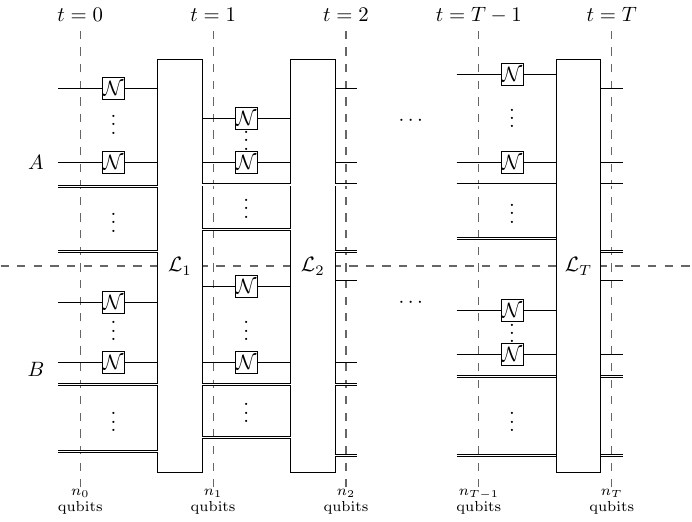}
    \caption{An illustration of $\mathrm{i.i.d._\cN}(C)$. In each time step, every qubit is subject to the noise channel $\cN$. The classical subsystems are assumed to be noise-free. The registers are  partitioned into disjoint subsets A and B such that the quantum channels $\cL_i$ are separable with respect to this bipartition.}
    \label{fig:figure2}
\end{figure}

\begin{lemma} \label{lem:chiSEPcontraction}
    For every non-unitary qubit channel $\cN$, there exists a constant $p\in (0,1]$ such that $\chisep{A:B}\br{\rho^{i+1}_{AB}} \leq \br{1-p^{n}} \chisep{A:B}\br{\rho^{i}_{AB}}$.
\end{lemma}

\begin{proof}
    In time step $i+1$, each of the $n_{i}$ qubits in the noisy circuit goes through the noise channel $\cN:\cM_2 \rightarrow \cM_2$ followed by a separable channel $\cL_{i+1}\in \mathrm{SepC}\br{A_{i},A_{i+1}:B_{i},B_{i+1}}$. We use two different approaches to prove Lemma~\ref{lem:chiSEPcontraction} for unital and non-unital quantum channels. 
    
    \textbf{Case I: $\cN$ is unital.} In this case, we prove the contractivity of $\chisep{A:B}$ independent of the initial value $\chisep{A:B}\br{\rho^{i}}$. By Lemma \ref{lem:UnitalCaseLemma}, we may write 
    \begin{equation}
        \cN= \br{1-p} \, \cT + p \, \cB \enspace,
    \end{equation}
    where $\cT:\cM_2\ra\cM_2$ is a unitary quantum channel and $\cB:\cM_2\ra \cM_2$ is entanglement breaking. Since $\cN$ itself is not unitary, we have $p>0$. There exists a separable quantum channel $\cS:\cM_{2^{n_{i}}} \rightarrow \cM_{2^{n_{i}}}$ with respect to the bipartition into $A$ and $B$, such that 
    \begin{equation}
        \cN^{\otimes n_{i}} = \br{1-p^{n_{i}}} \cS + p^{n_{i}}\, \cB^{\otimes n_i}\enspace,
    \end{equation}
    and we can compute
    \begin{align}
       \chisep{A:B}\br{\rho^{i+1}} &= \chisep{A:B}\br{\br{\cL_{i+1} \circ \br{\cN^{\otimes n_{i}} \otimes \cI^{A^c_iB^c_i}} }\br{\rho^{i}}} \\
       &\leq \chisep{A:B}\br{\br{\cN^{\otimes n_{i}} \otimes \cI^{A^c_iB^c_i}}\br{\rho^i}}\\
       &\leq \br{1-p^{n_{i}}} \chisep{A:B}\br{\br{\cS\otimes \cI^{A^c_iB^c_i}}\br{\rho^i}} \\
       &\hspace{1.4em} + p^{n_{i}}\, \chisep{A:B}\br{\br{\cB^{\otimes n_{i}}\otimes \cI^{A^c_iB^c_i}}\br{\rho^i}}\\
       &= \br{1-p^{n_i}} \chisep{A:B}\br{\br{\cS\otimes \cI^{A^c_iB^c_i}}\br{\rho^i}} \\
       &\leq \br{1-p^{n}} \chisep{A:B}\br{\rho^i}\enspace.
    \end{align}
    Here, the first and the last inequalities follow from monotonicity of $\chisep{A:B}$ under separable quantum channels, and the second inequality follows from convexity of $\chisep{A:B}$. Moreover, the second equality holds since $\cB$ is entanglement breaking and $\rho^i$ is classical on $A^c_iB^c_i$. 
    
    \textbf{Case II: $\cN$ is non-unital.} By Lemma~\ref{lem:Chi2ContractionFromEta} in the appendix for $\cT=\cN^{\otimes n_i}$ we have 
    \begin{equation} \label{eq:chisep_contraction}
        \chisep{A:B}\br{\br{\cN^{\otimes n_{i}} \otimes \cI^{A^c_iB^c_i}}\br{\rho^i}} \leq \sqrt{\eta_\chi\br{\cN^{\otimes n_i}}} \chisep{A:B} \br{\rho^i}\enspace.
    \end{equation}
    By Eq.~\eqref{eq:eta_chi vs eta_tr}, we have $\eta_\chi\br{\cN^{\otimes n_i}}\leq \eta_\mathrm{tr}\br{\cN^{\otimes n_i}}$. So we focus on bounding $\eta_\mathrm{tr}\br{\cN^{\otimes n_i}}$. The definition \eqref{equ:QDobrushin} of the trace-norm contraction coefficient might not be easily computable for all channels, and it is unclear how it behaves under tensor powers. In order to prove the lemma we prove a general upper bound on the trace-norm contraction coefficient of a quantum channel.
    
    \begin{lemma} \label{lem:UB on Dobrushin}
        For any quantum channel $\cT:\cM_d \rightarrow \cM_d$, we have
        \begin{equation}
            \eta_{\mathrm{tr}}(\cT)\leq \sqrt{1-\frac{\lambda^{\mathrm{out}}_{\min}\br{\cT^\dagger \circ \cT}}{d^2}}\enspace.
        \end{equation}
   Here, we used the minimal output eigenvalue $\lambda^{\mathrm{out}}_{\min} \br{\cT^\dagger \circ \cT} = \displaystyle{\min_{\psi,\phi}} \,\bra{\psi}(\cT^\dagger \circ \cT)(\proj{\phi}{\phi})\ket{\psi}$, where the minimum goes over pure quantum states.
    \label{lem:EtaBoundMinOutEV}
    \end{lemma}

    \begin{proof}
        In \cite{watrous2009semidefinite} the fidelity $\mathrm{F}(\rho,\sigma)=\|\sqrt{\rho}\sqrt{\sigma}\|^2_1$ of quantum states $\rho$ and $\sigma$ on $\C^d$ is expressed as a semidefinite program 
        \begin{equation}
            \mathrm{F}(\rho,\sigma) = \max\cbr{ \bra{v} W\ket{v} ~:~ W_{AE}\geq 0, \Tr_E\br{W_{AE}} \leq \rho },
        \end{equation}
        where $v\in \C^{d}\otimes \C^{d_E}$ is a purification of $\sigma$. Consider the case $\rho=\cT(\proj{\psi}{\psi})$ and $\sigma=\cT(\proj{\phi}{\phi})$ for pure states $\ket{\psi},\ket{\phi}\in\C^d$. Let $V:\C^{d}\rightarrow \C^{d}\otimes \C^{d_E}$ denote a Stinespring dilation of $\cT$ such that 
        \begin{equation}
            \cT(X) = \Tr_E\br{VXV^{\dagger}}\enspace.
        \end{equation}
        Note that $V\ket{\phi}\in \C^{d}\otimes \C^{d_E}$ is a purification of $\sigma$. Choosing $W=\cT(\proj{\psi}{\psi})\otimes \frac{\id_E}{d_E}$ we find that 
        \begin{equation}
             \mathrm{F}(\rho,\sigma)\geq \bra{\phi}V^\dagger \br{ \cT(\proj{\psi}{\psi})\otimes \frac{\id_E}{d_E}} V\ket{\phi} = \frac{1}{d_E}\bra{\phi}\cT^\dagger \circ \cT\br{\proj{\psi}{\psi}} \ket{\phi}\enspace,
        \end{equation}
        where we used that the adjoint of $\cT$ has the Stinespring dilation
        \begin{equation}
            \cT^\dagger(X) = V^\dagger (X\otimes \id_E) V\enspace.
        \end{equation}
        Finally, we can apply the well-known Fuchs-van-de-Graaf inequality relating the trace norm and the fidelity ($\| \rho - \sigma \|_1 \leq \sqrt{1-\mathrm{F}(\rho, \sigma)}$ for any quantum states $\rho, \sigma$) to conclude
        \begin{align*}
            \eta_{\text{tr}}(\cT) &= \frac{1}{2}\max_{\psi\perp\phi}\| \cT(\proj{\psi}{\psi}-\proj{\phi}{\phi})\|_1 \\
            &\leq \max_{\psi\perp\phi} \sqrt{1-\mathrm{F}(\cT(\proj{\psi}{\psi}),\cT(\proj{\phi}{\phi}))} \\
            &\leq \max_{\psi\perp\phi}\sqrt{1-\frac{1}{d^2}\bra{\phi}\br{\cT^\dagger \circ \cT} \br{\proj{\psi}{\psi}} \ket{\phi}} \\
            & \leq  \sqrt{1-\frac{\lambda^{\mathrm{out}}_{\min}\br{ \cT^\dagger\circ \cT}}{d^2}}\enspace,
        \end{align*}
        where we used that $d_E\leq d^2$ for any quantum channel $\cT:\cM_d \rightarrow \cM_d$.
    \end{proof}

Computing $\lambda^{\mathrm{out}}_{\min}\br{\cT^\dagger \circ \cT}$ might be difficult in general, but for our purposes an easy lower bound suffices. This bound also behaves nicely under tensor product of channels; an important property which we use below. For $\ket{\phi}=\sum_{i\in[d]} \alpha_i \ket{i}$, let $\ket{\overline{\phi}}=\sum_{i\in[d]} \overline{\alpha_i} \ket{i}$, where $\overline{\alpha_i}$ denotes the complex conjugate of $\alpha_i$. Then we have
\begin{equation}
\lambda^{\mathrm{out}}_{\min} \br{\cT^\dagger\circ \cT} = \min_{\psi,\phi} \,\bra{\psi}(\cT^\dagger \circ \cT)(\proj{\phi}{\phi})\ket{\psi}=\min_{\psi,\phi}\, \bra{\psi\otimes \overline{\phi}}C_{\cT^\dagger\circ \cT} \ket{\psi\otimes \overline{\phi}} \geq \lambda_{\min} \br{C_{\cT^\dagger\circ \cT}}.
\label{eq:EasyBound}
\end{equation}

  We are now in position to present the proof of the case where $\cN$ is a non-unital qubit channel. We may write 
    \begin{equation}
        \cN= (1-q)\,\cM' + q \, \cM ,
    \end{equation}
  where $q\in(0,1]$ and $\cM$ is a non-unital extreme point of the set of qubit channels. By the characterization of extreme points of the set of quantum channels (see preliminaries or \cite{Choi}), the non-unital extremal qubit channel $\cM$ admits a Kraus decomposition with two Kraus operators $K_1,K_2:\C^2\ra \C^2$ such that the operators $K_1^\dagger K_1,K_1^\dagger K_2,K_2^\dagger K_1,K_2^\dagger K_2$ are linearly independent. Since these operators are exactly the Kraus operators of the completely positive map $\cM^\dagger \circ \cM$, we conclude that the Choi matrix $C_{\cM^\dagger \circ \cM}$ is full-rank and that
\[
\lambda_{\min} \br{C_{\cN^\dagger \circ \cN}} \geq q^2\lambda_{\min} \br{C_{\cM^\dagger \circ \cM}}>0 .
\]  

Setting $\cT=\cN^{\otimes n_i}$ and using~\eqref{eq:EasyBound} and Lemma~\ref{lem:EtaBoundMinOutEV}, we have 
    \begin{align}
        \eta_{\text{tr}}(\cT) &\leq \sqrt{1-\frac{\lambda_{\min}\br{C_{\cT^\dagger \circ \cT}}}{2^{2n_i}}} \!=\!
        \sqrt{1-\frac{\lambda_{\min}\br{\br{C_{\cN^\dagger \circ \cN}}^{\otimes n_i}}}{2^{2n_i}}} \!=\! \sqrt{1-\br{\frac{\lambda_{\min}\br{C_{\cN^\dagger \circ \cN}}}{4}}^{n_i}} \nonumber\\
        &\leq \sqrt{1-\br{\frac{q^2\lambda_{\min}\br{C_{\cM^\dagger \circ \cM}}}{4}}^{n_i}} \leq 1-\br{\frac{q^2\lambda_{\min}\br{C_{\cM^\dagger \circ \cM}}}{8}}^{n_i}\enspace. \label{eq:ubd_on_eta_tr}
    \end{align}
    Combining Eqs.~\eqref{eq:chisep_contraction} and~\eqref{eq:ubd_on_eta_tr}, for $p=\br{q^2/16}\lambda_{\min}\br{C_{\cM^\dagger \otimes \cM}}>0$, we have 
    \begin{align*}
        \chisep{A:B}\br{\rho^{i+1}} &= \chisep{A:B}\br{\br{\cL_{i+1} \circ \br{\cN^{\otimes n_{i}} \otimes \cI^{A^c_iB^c_i}} }\br{\rho^{i}}}\\ 
        &\leq \chisep{A:B}\br{\br{\cN^{\otimes n_{i}} \otimes \cI^{A^c_iB^c_i}}\br{\rho^i}}\\
        &\leq \br{1-p^{n_i}} \chisep{A:B}\br{\rho^i} \leq \br{1-p^{n}} \chisep{A:B}\br{\rho^i}\enspace.
    \end{align*}
\end{proof}

One may ask whether analogous statements to Lemma \ref{lem:chiSEPcontraction} hold for entanglement measures constructed from other distance measures besides the $\chi^2$-divergence. Indeed, using results on contraction coefficients from~\cite{hirche2023quantum} (that appeared after the publication of the present paper), it is possible to prove an analogon of Lemma \ref{lem:chiSEPcontraction} for the relative entropy of entanglement. We do not know whether the same type of result is true for the trace distance to separable states. 

Recall the setup of Theorem~\ref{thm:no_long_memory}: Let $C$ be an arbitrary quantum circuit in which the gates are allowed to be arbitrary quantum channels acting on several qubits\suppress{, possibly controlled by a classical register of arbitrary (finite) dimension. Moreover, any classical processing is assumed to be perfect and instantaneous}. A supply of fresh ancilla qubits is available and the number of ancillas in $C$ in each time step is restricted by the circuit width. We will now consider two parallel instances of the circuit and apply them to an entangled state. If the parallel instances of the circuits are affected by the $\mathrm{i.i.d.}_\cN$ error model for a non-unitary qubit channel $\cN$, then they cannot keep entanglement for arbitrary long time. Theorem~\ref{thm:no_long_memory} will then follow from the following lemma:  

\begin{lemma}\cite[Theorem 3.56]{TQI_Watrous}\label{lem:tr_dist. vs dimond}
    Let $\cI_d: \cM_d \rightarrow \cM_d$ be the identity channel. Let $\cT: \cM_d \rightarrow \cM_d$ be a quantum channel satisfying $\onenorm{\cT-\cI_d} \leq \epsilon$, for some $\epsilon\in[0,2]$. Then for any $d'\geq 1$, we have
    \begin{equation*}
        \onenorm{\cT\otimes \cI_{d'}-\cI_{d}\otimes \cI_{d'}} \leq \sqrt{2\epsilon}.
    \end{equation*}
    This in particular implies that $\onenorm{\cT^{\otimes 2} -\cI_{d}^{\otimes 2}} \leq 2\onenorm{\cT\otimes \cI_{d}-\cI_{d}\otimes \cI_{d}} \leq 2\sqrt{2\epsilon}$.
\end{lemma}

Combining Lemmas~\ref{lem:chiSEPcontraction} and \ref{lem:tr_dist. vs dimond}, we have:

\setcounter{section}{1}
\setcounter{theorem}{0}

\begin{theorem}[\textbf{Restated}]
    Let $\cI$ denote the qubit identity channel and $\cN$ be a non-unitary qubit channel. Then there exists a constant $p\in(0,1]$ only depending on $\cN$ such that the following holds:
    Let $C$ be an arbitrary circuit of length $T$ and width $n$ such that $T\geq \br{2/p}^{2n}$. Then for any pair of ideal encoding and decoding maps $\cE$ and $\cD$ (of appropriate input and output dimensions), we have
    \begin{equation*}
        \onenorm{\cD \circ \mathrm{i.i.d.}_\cN\br{C} \circ \cE - \cI} \geq \epsilon_0 \enspace,
    \end{equation*}
    for some universal constant $\epsilon_0\geq 1/128$.
\end{theorem} 

\begin{proof}
    Consider two parallel instances of the circuit $C$. By definition of the $\mathrm{i.i.d.}$ error model, we have $\mathrm{i.i.d.}_\cN\br{ C^{\otimes 2}}=\br{\mathrm{i.i.d.}_\cN\br{C}}^{\otimes 2}$. Note that $C^{\otimes 2}$ is a circuit of width $2n$ and length $T$, with $A$ being the registers of the first copy of $C$ and $B$ being the registers of the second copy. By Lemma~\ref{lem:expr_chisep} in the appendix, for any input state $\rho^0$, we have $\chisep{A:B}\br{\rho^0}\leq 2^{2n}$. Therefore, by Lemma~\ref{lem:chiSEPcontraction}, for $T\geq (2/p)^{2n}$, it is easy to see that $\chisep{A:B}\br{\br{\mathrm{i.i.d.}_\cN\br{ C}}^{\otimes 2}(\rho^0)}\leq 1/16$. In particular, for any encoding map $\cE$ and any 2-qubit state $\sigma$, we have $\chisep{A:B}\br{\br{\br{\mathrm{i.i.d.}_\cN\br{ C}}^{\otimes 2}\circ\cE^{\otimes 2}}(\sigma)}\leq 1/16$. For any decoding map $\cD$, since $\cD^{\otimes 2}$ is a separable channel, we have \begin{equation*}
        \chisep{A:B}\br{\br{\cD^{\otimes 2}\circ \br{\mathrm{i.i.d.}_\cN\br{ C}}^{\otimes 2}\circ\cE^{\otimes 2}}(\sigma)}\leq \chisep{A:B}\br{\br{\mathrm{i.i.d.}_\cN\br{ C}}^{\otimes 2}\circ\cE^{\otimes 2}}(\sigma)\leq \frac{1}{16}.
    \end{equation*}
    Let $\cM\coloneqq \cD \circ \mathrm{i.i.d.}_\cN\br{C} \circ \cE$. Then we have
    \begin{equation*}
        \onenorm{\cM^{\otimes 2}-\cI^{\otimes 2}}=  \sup_\sigma \onenorm{\cM^{\otimes 2}(\sigma)-\sigma} \geq \sup_\sigma \; \br{\dsep{A:B}(\sigma)-\dsep{A:B}(\cM^{\otimes 2}(\sigma))}\enspace,
    \end{equation*}
    where $\dsep{A:B}$ denotes the 1-norm distance from the set of separable states $\mathrm{Sep}\br{A:B}$ defined for a quantum state $\tau_{AB}$ as $\dsep{A:B}\br{\tau_{AB}} \coloneqq \min_{\zeta\in \mathrm{Sep}\br{A:B}} \onenorm{\tau-\zeta}$. For $\sigma$ equal to the maximally entangled state, it is straightforward to see that $\dsep{A:B}(\sigma)\geq 1/2$.  Moreover, by Eq.~\eqref{eq:tr_dist vs chi2}, we have $\dsep{A:B}\br{\cM^{\otimes 2}(\sigma)}\leq \sqrt{1/16}=1/4$, for any 2-qubit state $\sigma$. Hence, we have $\onenorm{\cM^{\otimes 2}-\cI^{\otimes 2}}\geq 1/4$. The statement of the theorem follows from Lemma~\ref{lem:tr_dist. vs dimond} for $\epsilon_0\geq1/128$.
\end{proof}
\setcounter{section}{3}

We should point out a connection between our proof of Theorem \ref{thm:no_long_memory} (including Lemma \ref{lem:UB on Dobrushin}) and the zero-error classical capacity of a quantum channel~\cite{medeiros2005quantum,beigi2007complexity}. The zero-error classical capacity $C_0(\cN)$ of a quantum channel $\cN$ arises as the regularization of the logarithm of the (quantum) \emph{independence number} $\alpha(\cN)$ which equals the maximal number of orthogonal pure states that are mapped to orthogonal outputs by the channel $\cN$ (see~\cite{duan2012zero} for details). It is clear from this definition that $\alpha(\cN)=1$ if and only if $\eta_{\text{tr}}(\cN)<1$. Qualitatively, Theorem \ref{thm:no_long_memory} says that a finite-size quantum memory affected by a non-unitary qubit noise channel $\cN$ cannot store information for arbitrary long times. Using our proof technique, this statement can be generalized beyond qubit noise channels as follows: If $\cN$ arises as a non-trivial convex combination with an entanglement breaking quantum channel or satisfies $\alpha(\cN^{\otimes k})=1$, for some $k\in\N$, then no quantum memory of width $k$ affected by noise $\cN$ can store information for arbitrary long times. In particular, if $C_0(\cN)=0$, then no finite width will allow for storing information for arbitrary long times. The zero-error classical capacity can be superactivated in the sense that there are quantum channels $\cN$ for which $\alpha(\cN)=1$ but $\alpha(\cN^{\otimes k})>1$ for some $k\in \N$ (see~\cite{duan2009super,cubitt2011superactivation}), and it would be interesting to study the maximal storage times of quantum memories for such examples. Our techniques in Lemma \ref{lem:UB on Dobrushin} can be used to show more generally that such a superactivation does not happen for channels of maximal Kraus rank. We remark that the impossibility of superactivation of the zero-error classical capacity has been shown previously in~\cite{park2012zero,shirokov2015superactivation}, for qubit channels. However, unlike Lemma \ref{lem:UB on Dobrushin} these results did not give any quantitative estimate on the trace-norm contraction coefficient and therefore they do not imply our main quantitative bound.

\section{Implications for fault-tolerant quantum computation} \label{sec:implications}

\subsection{A high-level definition of fault tolerance}

A \emph{quantum fault tolerance scheme} is a circuit encoding $FT$ which maps any ideal quantum circuit $C$ to a family of its \emph{fault-tolerant simulation} circuits $FT^\ell(C)\coloneqq FT\br{\ell,C}$. As defined formally below, the circuit encoding is designed to allow the implementation of the circuit $C$ even with faulty components.

Let $T=Length(C)$ denote the length of the circuit $C$. We denote by $n_t(C)$ the number of (logical) qubits in the circuit $C$ after $t\in\{0,1,\ldots,T\}$ time steps and use $n=Width(C)$ to denote the width of $C$, i.e., $n=\max_t n_t(C)$. Similarly, let $n'=Width(FT^\ell(C))$ be the width, i.e, the maximum number of (physical) qubits and $T'=Length\br{FT^\ell(C)}$ be the length of $FT^\ell(C)$. Note that it is important to distinguish between the classical and quantum subsystems, as in our circuit model only quantum subsystems are subject to the noise. For the input circuit $C$, we may assume without loss of generality that there are no classical inputs, as we can consider the circuit obtained by conditioning on the state of the classical input systems. 

\setcounter{section}{4}
\setcounter{theorem}{7}
\begin{definition} \label{def:FT}
    A circuit encoding $FT$ is a fault tolerance scheme against an error model $\sE$ if there exist a family of ideal encoding maps $\cbr{\cE^\ell_n}$ and decoding maps $\cbr{\cD^\ell_n}$ only depending on $FT$ such that for every $\epsilon\in(0,2]$ and every quantum circuit $C$, there exists $\ell_0$ such that
    \begin{equation} \label{eq:FT-condition}
        \onenorm{ \cD^\ell_{n_T(C)} \circ \sE\br{FT^\ell(C)} \circ \cE^\ell_{n_0(C)} - C} \leq \epsilon \,, \quad \forall \ell\geq \ell_0 \enspace.
    \end{equation}
\end{definition} 
\setcounter{section}{1}
\setcounter{theorem}{1}
In the above definition, if the input Hilbert space of $C$ is trivial (e.g., if $C$ is a circuit for classical computation), then for every $\ell$, the circuit $FT^\ell(C)$ is also naturally defined to have no input, i.e., $\cE^\ell_0$ is fixed to be the identity map on $\C$. 

Formal definitions of fault tolerance are indeed invaluable in studying general properties and limitations of fault tolerance schemes. We emphasize that Definition~\ref{def:FT} is far from capturing all the different aspects of quantum fault tolerance. Indeed, there are several different ways of defining fault tolerance, each with their own advantages and limitations. Here, we aim at providing a high-level definition with minimal assumptions which is easy to work with and applies to most known and potential schemes. Below we discuss some examples of known fault tolerance protocols and explain how they fit into our definition.

Perhaps the most well-known family of fault tolerance schemes are the schemes based on concatenation. These schemes typically involve a quantum error correcting code $Q$ together with fault-tolerant gadgets for the elementary locations, i.e., state preparation gadgets, gate gadgets for a universal set of gates, storage gadgets, and measurement gadgets, in addition to error correction gadgets. For a quantum circuit $C$, the circuit $FT(C)$ is obtained by replacing the locations in $C$ by the corresponding gadgets, interlaced with error correction gadgets on the quantum registers, which now correspond to qubits of $C$ encoded in $Q$. The gadgets are designed to satisfy some fault tolerance properties, such that they not only act correctly in the absence of faults, but also control the propagation of errors if the number of errors is not too high. For $\ell\geq 1$, let $FT^\ell(C)=FT\br{FT^{\ell-1}(C)}$, where we define $FT^0(C)=C$. For any such scheme, assuming the distance between the noise channel $\cN$ and the identity channel is below a threshold, the output state of $\mathrm{i.i.d.}_\cN\br{FT^\ell(C)}$ is arbitrarily close to an encoded version of the output state of $C$ for a sufficiently large $\ell$ \cite{AGP05}.

Fault tolerance schemes based on topological codes are another family of schemes which have a threshold noise level below which arbitrarily long quantum computations are possible~\cite{denis02,WANG03,Raussendorf07}. For the surface code~\cite{Kitaev03} (and more generally, low-density parity check codes), a single syndrome bit measurement only involves a few neighboring qubits. This property allows error correction to be performed more efficiently for these codes~\cite{terhal15,Brown16}. Moreover, topological properties of surface codes allow for simpler and more efficient implementation of fault-tolerant gadgets for elementary locations~\cite{denis02,Campbell17,Horsman12,Brown17}. This in turn results in a desirable scaling of the overhead and reduction of the logical error rate when the size of the underlying surface code is increased. In this case, the index $\ell$ in definition~\ref{def:FT} corresponds to the size of the surface code. 

\subsection{A lower bound on the space overhead for fault tolerance} 

For a fault-tolerant simulation circuit $FT^\ell(C)$, the ratios $n'/n$ and $T'/T$ are respectively called the space overhead and the time overhead of $FT^\ell(C)$. We view a fault tolerance scheme as a compiler which reads the instructions given by the input circuit $C$ and encodes them into the fault-tolerant simulation circuits $FT^\ell(C)$, without optimizing the circuit $C$. This should be true even if the input circuit involves unnecessary steps that are completely independent of the output of the circuit. While in Eq.~\eqref{eq:FT-condition}, we only require the circuit $FT^\ell(C)$ to (approximately) implement the \emph{overall} logical operation of the circuit $C$ on a code space, this ``faithfulness'' property is in fact necessary to have a meaningful notion of overhead for a fault tolerance scheme. As a result of this property, we may assume $n'\geq n$ and $T'\geq T$, for any input circuit $C$ and any fault-tolerant simulation circuit $FT^\ell(C)$. The faithfulness property also rules out circuit encodings which map quantum circuits for classical computation to fully classical (hence, noiseless) circuits.  

Our result of the previous section can be used to prove a lower bound on the space overhead of any universal fault tolerance scheme satisfying Definition~\ref{def:FT} and the natural requirement that the overhead is at least $1$ for any fault-tolerant simulation circuit. 

\setcounter{section}{1}
\setcounter{theorem}{1}

\begin{theorem}[\textbf{Restated}]
    Let $\cN$ be a non-unitary qubit channel. For any fault tolerance scheme against the $\mathrm{i.i.d.}_\cN$ noise model, the number of physical qubits is at least 
    \begin{equation*}
        \max\cbr{\mathrm{Q}(\cN)^{-1}n,\alpha_\cN \log T}\enspace,
    \end{equation*}
    for circuits of length $T$ and width $n$, and any constant accuracy $\epsilon\leq 1/128$. Here, $\mathrm{Q}(\cN)$ denotes the quantum capacity of $\cN$ and $\alpha_\cN=\frac{1}{2\log 2/p}$, where $p\in(0,1]$ is a constant only depending on the channel $\cN$. In particular, fault tolerance is not possible even for a single time step if $\cN$ has zero quantum capacity or the space overhead is strictly less than $\mathrm{Q}(\cN)^{-1}$.
\end{theorem}

\begin{proof}
    Let $C$ be a quantum circuit of length $T$ on $n$ qubits in which in every time step an identity gate is applied on every qubit. Consider a fault tolerance scheme $FT$ satisfying Definition~\ref{def:FT}. For any constant $\epsilon\in (0,1/128)$, let $C'(\epsilon)$ be a fault-tolerant encoding of $C$ given by $FT$ satisfying Eq.~\eqref{eq:FT-condition}, for a pair of encoding and decoding maps $\br{\cE_\epsilon,\cD_\epsilon}$. Then by Theorem~\ref{thm:no_long_memory}, the width of $C'(\epsilon)$ is bounded in terms of its depth as $n'\geq \frac{\log T'}{2\log 2/p}$. Since $T'\geq T$ the lower bound of $\alpha_\cN \log T$ follows for $\alpha_\cN=\frac{1}{2\log 2/p}$. 
    
    For the second lower bound, consider the quantum channel obtained by removing the last layer of noise from $\mathrm{i.i.d.}_\cN\br{C'(\epsilon)}$ and let $\cE'_\epsilon$ denote the composition of $\cE_\epsilon$ with this channel. Then for the family of encoding and decoding maps $\br{\cE'_\epsilon,\cD_\epsilon}$, the overall channel $\cT_\epsilon=\cD_\epsilon\circ\cN^{\otimes n_\epsilon}\circ \cE'_\epsilon$ approximates the qubit identity channel up to arbitrarily high accuracy. But continuity of the quantum capacity~\cite{Leung09} implies that, for sufficiently small $\epsilon$, the channel $\cT_\epsilon$ (and hence $\cN$) has a non-zero quantum capacity. Moreover, the existence of the family of encoding and decoding maps $\br{\cE'_\epsilon,\cD_\epsilon}$ implies that, for the channel $\cN$, the inverse of the space overhead is an $\epsilon$-achievable communication rate for any arbitrarily small constant $\epsilon$ (See Definition~\ref{def:Qcapacity}). Therefore, by the definition of quantum capacity, we have $n'\geq \frac{1}{\mathrm{Q}(\cN)}n$. Since assistance by noiseless forward classical communication does not increase the quantum capacity of a channels~\cite{BKN2000,BDSW96}, the lower bound of $\mathrm{Q}(\cN)^{-1}n$ is still valid in our circuit model with free and noiseless classical computation.
    Finally, note that the above argument is valid even if there is a single time step of noisy operations. Therefore, fault tolerance is not achievable even for a single time step if $\cN$ has zero quantum capacity, or with an space overhead strictly less than $1/\mathrm{Q}(\cN)$.
\end{proof}
\setcounter{section}{4}
\setcounter{theorem}{8}

\bibliography{biblio.bib}

\appendix
\section{Appendix}

\begin{lemma} \label{lem:expr_chisep}
    Let $\rho_{XYAB}$ be a cc-qq state of the form $\rho_{XYAB} = \sum_{x,y} p_{xy} \proj{x}{x}_{X} \otimes \proj{y}{y}_{Y} \otimes \rho_{AB}^{xy}$. Then, we have
    \begin{align}
    \label{eq:bound_chi2_sep}
        \chisep{XA:BY}(\rho_{XYAB}) \leq d_{A} d_{B} - 1 \ .
    \end{align}
    In addition, we can express the $\chi^2$ divergence to separable states $\mathrm{Sep}\br{XA:BY}$ as
    \begin{align}
    \label{eq:expr_chi2_sep}
        \chisep{XA:BY}(\rho_{XYAB}) =  \left(\sum_{x,y} p_{xy} \sqrt{\chisep{A:B}(\rho^{xy}_{AB}) + 1} \right)^2 - 1 \ .
    \end{align}
\end{lemma}

\begin{proof}
    Consider the separable state $\sigma_{XYAB} = \sum_{x,y} p_{xy} \proj{x}{x}_{X} \otimes \proj{y}{y}_{Y} \otimes \frac{\id_{AB}}{d_A d_B}$. We get $\chi^2(\rho, \sigma) = \sum_{x,y} p_{xy} d_{A} d_{B} \Tr\br{ \br{\rho^{xy}}^2 } - 1 \leq d_A d_B - 1$. This shows~\eqref{eq:bound_chi2_sep}.

    To show Eq.~\eqref{eq:expr_chi2_sep}, by the data processing inequality for the $\chi^2$ divergence, we have that 
    \begin{align*}
        \chisep{XA:BY}(\rho_{XYAB}) = \min_{\substack{q_{xy} \\ \sigma^{xy} \in \mathrm{Sep}(A:B)}} \chi^2\left(\rho_{XYAB}, \sum_{x,y} q_{xy} \proj{x}{x}_{X} \otimes \proj{y}{y}_{Y} \otimes \sigma^{xy}_{AB} \right) \ ,
    \end{align*}
    where $q_{xy}$ is a probability distribution. Note that for any $x,y$ for which $q_{xy} > 0$, $\sigma^{xy}_{AB}$ has to be separable if $\sum_{x,y} q_{xy} \proj{x}{x} \otimes \proj{y}{y} \otimes \sigma^{xy}_{AB}$ is separable. We can write
    \begin{align*}
       \min_{\substack{q_{xy} \\ \sigma^{xy} \in \mathrm{Sep}(A:B)}} &\chi^2\!\br{\rho_{XYAB}, \sum_{x,y} q_{xy} \proj{x}{x}_{X} \otimes \proj{y}{y}_{Y} \otimes \sigma_{AB}^{xy}} \\
      &= \br{\min_{q_{xy}} \sum_{x,y} p_{xy}^2 q_{xy}^{-1} \min_{\sigma^{xy}_{AB} \in \mathrm{Sep}(A:B)} \Tr\br{ \rho^{xy} \br{\sigma^{xy}}^{-1/2} \rho^{xy} \br{\sigma^{xy}}^{-1/2}} } - 1 \\
        &= \left(\min_{q_{xy}} \sum_{x,y} p_{xy}^2 q_{xy}^{-1} \br{\chi^2_{\mathrm{Sep}}(\rho^{xy}_{AB}) + 1} \right) - 1 \ .
    \end{align*}
    By convexity of the function $z \mapsto z^{-1}$ and introducing $Z = \sum_{x,y} \sqrt{p_{xy}^2 \br{\chisep{A:B}(\rho^{xy}_{AB}) + 1}}$, we have
    \begin{align*}
       \sum_{x,y} \frac{\sqrt{p_{xy}^2 \br{\chisep{A:B}(\rho^{xy}_{AB}) + 1}}}{Z} \frac{\sqrt{p_{xy}^2 \br{\chisep{A:B}(\rho^{xy}_{AB}) + 1}}}{q_{xy}} 
       &\geq \frac{1}{\sum_{x,y} \frac{1}{Z} q_{xy}}=  Z \ .
    \end{align*}
    We also have equality by choosing $q_{xy} = \frac{\sqrt{p_{xy}^2 \br{\chisep{A:B}(\rho^{xy}_{AB})+1}}}{Z}$ and we obtain
    \begin{align*}
       \chisep{XA:BY}(\rho_{XYAB}) = Z^2 - 1\enspace.
    \end{align*}
    
\end{proof}

\begin{lemma}\label{lem:Chi2ContractionFromEta}
Let $\rho_{X Y A_1 B_1}$ be a cc-qq state of the form given in Lemma~\ref{lem:expr_chisep}. Then, for any separable quantum channel $\cT\in \mathrm{SepC}\br{(A_1,A_2):(B_1,B_2)}$, we have
\[
\chisep{XA_2:B_2Y}((\cI_{XY}\otimes \cT)(\rho_{XYA_1B_1})) \leq \sqrt{\eta_{\chi}(\cT)}\chisep{XA_1:B_1Y}(\rho_{XYA_1B_1}).
\]
\end{lemma}

\begin{proof}
By Lemma \ref{lem:expr_chisep} we have
\begin{align*}
\chisep{XA_2:B_2Y}((\cI_{XY}\otimes \cT)(\rho_{XYA_1B_1})) &= \left(\sum_{x,y}p_{xy}\sqrt{\chisep{A_2:B_2}(\cT(\rho^{xy}_{A_1B_1}))+1}\right)^2 - 1 \\
&\leq \left(\sum_{x,y}p_{xy}\sqrt{\eta_{\chi}(\cT)\chisep{A_1:B_1}(\rho^{xy}_{A_1B_1})+1}\right)^2 - 1
\end{align*}
The final expression can be written as
\begin{align*}
&\lb\lbr \sum_{x,y}p_{xy}\sqrt{\eta_{\chi}(\cT)\chisep{A_1:B_1}(\rho^{xy}_{A_1B_1})+1}\rbr - 1\rb\lb\lbr\sum_{x,y}p_{xy}\sqrt{\eta_{\chi}(\cT)\chisep{A_1:B_1}(\rho^{xy}_{A_1B_1})+1}\rbr + 1\rb \\
&=\lb \sum_{x,y}p_{xy}\lbr\sqrt{\eta_{\chi}(\cT)\chisep{A_1:B_1}(\rho^{xy}_{A_1B_1})+1}-1\rbr\rb\lb\sum_{x,y}p_{xy}\lbr\sqrt{\eta_{\chi}(\cT)\chisep{A_1:B_1}(\rho^{xy}_{A_1B_1})+1}+1\rbr\rb \\
&\leq \sqrt{\eta_{\chi}(\cT)}\lb\sum_{x,y}p_{xy}\lbr\sqrt{\chisep{A_1:B_1}(\rho^{xy}_{A_1B_1})+1}-1\rbr\rb\lb\sum_{x,y}p_{xy}\lbr\sqrt{\chisep{A_1:B_1}(\rho^{xy}_{A_1B_1})+1}+1\rbr\rb \\
&=\sqrt{\eta_{\chi}(\cT)} \lb\lb\sum_{x,y}p_{xy}\sqrt{\chisep{A_1:B_1}(\rho^{xy}_{A_1B_1})+1}\rb^2 - 1\rb \\
&=\sqrt{\eta_{\chi}(\cT)}\chisep{XA_1:B_1Y}(\rho_{XYA_1B_1}),
\end{align*}
where we used that $\eta_{\chi}(\cT)\leq 1$ and the elementary inequality $\sqrt{at+1}-1\leq \sqrt{a}\lb \sqrt{t+1}-1\rb$, which is valid for $a\in \left[ 0,1\right]$ and any $t\geq 0$.
\end{proof}

\end{document}